\pdfoutput=1
\documentclass{article}  % Comment this line out if you need a4paper
\usepackage{arxiv}
\usepackage{amsthm}
\usepackage{flushend}

\usepackage[usenames, dvipsnames]{color}

\usepackage[pdfa,hidelinks]{hyperref}
\usepackage{optidef}
\usepackage{graphics} % for pdf, bitmapped graphics files
\graphicspath{{figures/}}
\usepackage{amsmath} % assumes amsmath package installed
\usepackage{amssymb}  % assumes amsmath package installed

\usepackage{verbatim}

\newcounter{rmnum}
\newenvironment{romannum}{\begin{list}{{\upshape (\roman{rmnum})}}{\usecounter{rmnum}
\setlength{\leftmargin}{2pt}
\setlength{\rightmargin}{4pt}
\setlength{\itemsep}{1pt}
\setlength{\itemindent}{5pt}
}}{\end{list}}

\def\Ebox#1#2{%
\begin{center} 
\includegraphics[width= #1\hsize]{#2}\end{center}}

\def\bl#1{{\color{blue}#1}}
\def\rd#1{{\color{red}#1}}

\newlength{\noteWidth}
\long\def\notes#1{\ifinner
             {\tiny #1}
             \else
             \marginpar{\parbox[t]{\noteWidth}{\raggedright\tiny #1}}
             \fi}
             
                       \def\notes#1{}
\def\archive#1{}

\def\spm#1{\notes{spm:  #1}}
\def\rwm#1{\notes{\bl{rwm:  #1}}}

\def\jjm#1{\notes{\rd{jjm:  #1}}}

\def\state{{\sf X}}

\def\MC{\text{\MC}}

\def\transpose{{\hbox{\it\tiny T}}}

\def\argmin{\mathop{\rm arg\, min}}

\def\argmax{\mathop{\rm arg\, max}}

\newcommand{\field}[1]{\mathbb{#1}}

\def\Re{\field{R}}

\def\eqdef{\mathbin{:=}}

\newtheorem{theorem}{Theorem}[section]
\newtheorem{corollary}[theorem]{Corollary}
\newtheorem{proposition}[theorem]{Proposition}
\newtheorem{lemma}[theorem]{Lemma}

\usepackage{cleveref}
\Crefname{corollary}{Corollary}{Corollaries}
\Crefname{eqnarray}{eq.}{eqs.}
\Crefname{equation}{eq.}{eqs.}

\Crefname{figure}{Fig.}{Figs.}
\Crefname{tabular}{Tab.}{Tabs.}
\Crefname{table}{Tab.}{Tabs.}
\Crefname{lemma}{Lemma}{Lemmas}

\Crefname{theorem}{Thm.}{Thms.}
\Crefname{definition}{Definition}{Definitions}
\Crefname{section}{Section}{Sections}
\Crefname{proposition}{Prop.}{Propositions}
\Crefname{assumption}{Assumption}{Assumptions}
\Crefname{example}{Example}{Examples}

\def\Theorem#1{Theorem~\ref{#1}}

\def\Section#1{Section~\ref{#1}}

\def\barell{{\overline {\ell}}}

\def\bfmath#1{{\mathchoice{\mbox{\boldmath$#1$}}%
{\mbox{\boldmath$#1$}}%
{\mbox{\boldmath$\scriptstyle#1$}}%
{\mbox{\boldmath$\scriptscriptstyle#1$}}}}

\def\bfmg{\bfmath{g}}

\def\bfmx{\bfmath{x}}

\def\bfmz{\bfmath{z}}

\def\bfnu{\bfmath{u}}

\def\bfgamma{\bfmath{\gamma}} 
\def\bfvarrho{\bfmath{\varrho}}

\def\bflambda{\bfmath{\lambda}}
\def\bfbeta{\bfmath{\beta}}

\def\bfmg{g}

\def\bfmx{x}
\def\bfmz{z}

\def\bfnu{u}   

\def\bfvarrho{\varrho}

\def\bfgamma{\gamma}
\def\bfvarrho{\varrho}

\def\bflambda{\lambda}
\def\bfbeta{\beta}

\def\bfmell{\bfmath{\ell}}
\def\bfmell{\ell}
 \def\FRAC#1#2#3{\genfrac{}{}{}{#1}{#2}{#3}}

\def\ddt{{\mathchoice{\FRAC{1}{d}{dt}}%
{\FRAC{1}{d}{dt}}%
{\FRAC{3}{d}{dt}}%
{\FRAC{3}{d}{dt}}}}

\def\dddt{{\mathchoice{\FRAC{1}{d^2}{dt^2}}%
		{\FRAC{1}{d^2}{dt^2}}%
		{\FRAC{3}{d^2}{dt^2}}%
		{\FRAC{3}{d^2}{dt^2}}}}

\def\ddzp{{\mathchoice{\FRAC{1}{\partial }{\partial z}}%
		{\FRAC{1}{\partial }{\partial z}}%
		{\FRAC{3}{\partial }{\partial z}}%
		{\FRAC{3}{\partial }{\partial z}}}}
	
	\def\ddxp{{\mathchoice{\FRAC{1}{\partial }{\partial x}}%
			{\FRAC{1}{\partial }{\partial x}}%
			{\FRAC{3}{\partial }{\partial x}}%
			{\FRAC{3}{\partial }{\partial x}}}}

\def\ddnup{{\mathchoice{\FRAC{1}{\partial }{\partial u}}%
		{\FRAC{1}{\partial }{\partial u}}%
		{\FRAC{3}{\partial }{\partial u}}%
		{\FRAC{3}{\partial }{\partial u}}}}

\def\clL{{\cal L}}

\def\clT{{\cal T}}
\def\clU{{\cal U}}

\def\cX{c_{\text{\tiny X}}}
\def\cG{c_{\text{g}}}
\def\cdG{c_{{\text{\lower1pt\hbox{d}}}} }

\def\As#1{\noindent
	\textbf{(#1)}  \ \ }

\def\MCav{\text{MC}_g^\text{avg}}
\def\MVav{\text{MV}_i^\text{avg}}
\def\varrhoav{\varrho^\text{avg}}

\usepackage{bbold} 

\RenewEnviron{BaseMiniExclam}[7]{%
	\selectConstraintMult{#1}%
	\renewcommand{\localOptimalVariable}{#2}%
	\begin{subequations}
		\ifthenelse{\equal{#7}{b}}{\allowdisplaybreaks}%
		#4
		\begin{alignat}{5}
		\bodyobj{#2}{#3}{#6}{#5}
		\BODY
		\end{alignat}
	\end{subequations}%
	\setStandardMini
}

\makeatletter
\def\thanks#1{\protected@xdef\@thanks{\@thanks
		\protect\footnotetext{#1}}}
\makeatother
\title{\LARGE \bf
State Space Collapse in
 \\
 Resource Allocation for Demand Dispatch
}

\author{Joel Mathias \\ Dept. ECE \\ Univ. of Florida\\
	\And Robert Moye \\ Dept. ECE, Univ. of Florida\\Rainbow Energy Marketing Corp.\\
	\And Sean Meyn \\ Dept. ECE\\ Univ. of Florida \\
	\And Joseph Warrington \\ Automatic Control Laboratory \\ETH Zurich
\thanks{$^*$This is  a preprint of the conference paper  \cite{matmoymey19},  to appear in the IEEE CDC,  December 2019.
Funding from the National Science Foundation under awards EPCN 1609131 \&\ CPS~1646229 is gratefully acknowledged.   Thanks also  to funding from  the State of Florida,  through a REET (Renewable Energy and Energy Efficiency Technologies) grant. Many thanks to Prof.\ Frank Kelly who suggested we survey the history in telecommunications economics to investigate parallels with the power industry. Simon's Institute, Berkeley, is gratefully acknowledged for hosting SM and JW in Spring 2018:
discussions during  the semester long program on \textit{Real Time Decision Making} served as   inspiration for this paper.  }% <-this % stops a space
\date{\vspace{-5ex}}
}

\begin{document}

\newcommand*{\QED}{\hfill\ensuremath{\blacksquare}}
\def\qedIEEE{\nobreak\hspace*{\fill}~\QED\par \unskip }
\def\ProofOf#1{\smallbreak\noindent\textit{#1:} \  }

\maketitle
\thispagestyle{empty}
%\pagestyle{empty}

%%%%%%%%%%%%%%%%%%%%%%%%%%%%%%%%%%%%%%%%%%%%%%%%%%%%%%%%%%%%%%%%%%%%%%%%%%%%%%%%
\begin{abstract}

\textit{Demand dispatch} is the science of extracting virtual energy storage through the automatic control of deferrable loads to provide balancing or regulation services to the grid, while maintaining consumer-end quality of service.

%\jlm{Modified to remove "our CDC paper" which sounded a little pushy. Also modifying to reduce abstract size.}
The control of a large collection of heterogeneous loads is in part a resource allocation problem, since different classes of loads are more valuable for different services. 

% Commenting below as SoC discussion is not needed in abstract. We have it in intro. 
% In prior research,   the evolution of QoS for  each class of loads  is modeled via a \textit{state of charge} surrogate, which is a part of the leaky battery model for the load classes.  

The goal of this paper is to unveil the structure of the optimal solution to the resource allocation problem, 
and investigate short-term market implications. It is found that the marginal cost for each load class evolves in a  \textit{two-dimensional subspace}: spanned by a co-state process and its derivative.
%The following conclusions are obtained:
%\begin{romannum}
	%\item  Optimal power deviation for each of the $M\ge2$ load classes evolves in a two-dimensional manifold.
	
	%\item  Marginal cost for each load class evolves in a two-dimensional \textit{subspace}: spanned by a co-state process and its derivative. 
%\end{romannum}		 

The resource allocation problem is recast to construct a \textit{dynamic competitive equilibrium model}, in which the consumer utility is the negative of the cost of deviation from ideal QoS. It is found that a competitive equilibrium exists with the equilibrium price equal to the negative of an optimal co-state process. Moreover, the equilibrium price is different than what would be obtained based on the standard assumption that the consumer's utility is a function of power consumption.   
% \jlm{Can we make this precise? In what way are the prices different?  }

% Joel: Commenting this as we have a discussion on the nature of prices in the introduction. Further, the reviewers point out the length of the abstract
%It is argued that price signals are not useful for control of the grid since they are inherently open loop.  However, the analysis may inform the creation of heuristics for payments within the context of contracts for services with consumers.

\end{abstract}

%%%%%%%%%%%%%%%%%%%%%%%%%%%%%%%%%%%%%%%%%%%%%%%%%%%%%%%%%%%%%%%%%%%%%%%%%%%%%%%%

 %%%%%%%%%%%%%%%%%%%%%%%%%%%%%%%%%%%%%%%%%%%%%%%%%%%%%%%%%%%%%%%%%%%%%%%%%%%%%%%%

\section{Introduction}
\label{s:intro}

The goals of this paper are twofold:  1)  to analyze the structure of the optimal solution to the resource allocation problem investigated in  \cite{cammatkiebusmey18},  and 2) to develop an understanding of the potential implications to market design.  

\subsection{Control techniques for demand dispatch }

The term \textit{demand dispatch} refers to the creation of virtual energy storage from deferrable loads.   The key to success is automation: an appropriate distributed control architecture ensures that bounds on quality of service (QoS) are met and simultaneously ensures that the loads provide aggregate grid services comparable to a large battery system.

The 2018 IMA volume on the control of energy markets
and grids contains several papers surveying distributed control techniques for demand dispatch 
\cite{chehasmatbusmey18,cheche17b,almesphinfropauami18}.    The present work and \cite{cammatkiebusmey18} are based on the results surveyed in \cite{chehasmatbusmey18}:   through distributed control, a heterogeneous population of loads such as residential water heaters can be controlled in such a way that  quality of service (QoS)---in terms of temperature and cycling---of each load obeys strict constraints,  while the power deviation can be adjusted up and down to provide grid services much like a battery system.

For loads such as water heating,  air-conditioning, and refrigeration (examples of thermostatically controlled loads, or TCLs), there is a natural analog with batteries, with thermal energy storage replacing electrical energy storage.   This storage is what allows large deviations in electric power consumption of the fleet, with imperceptible impact on service to consumers.  Ample evidence of this potential is presented in  \cite{cammatkiebusmey18, chehasmatbusmey18} and their references.  

The next question concerns the management of a large heterogeneous population of loads.   When heterogeneity is not large (e.g., each load in the collection is a residential water heater, but the size varies across the population),  then additional local control at each load can be designed to make the population appear homogeneous \cite{matbusmey17,matkadbusmey16}.   The problem addressed in this paper and in   \cite{cammatkiebusmey18} is control of a highly diverse population of loads.  Along with residential and commercial TCLs,  these might include the residential pool pumps in a region, along with water pumping for irrigation or waste management.  Control of the fleet is a dynamic resource allocation problem, which is formulated as a linear program in \cite{trotinstr16} and as a convex program in  \cite{cammatkiebusmey18}.

The convex program is revisited here.  The goal is to gain insight on the structure of the optimal solution for a model with one source of traditional generation  and grid services obtained through demand dispatch of a large population of loads.  It is assumed that there are $M$ classes of loads;  in each class, the population is homogeneous.   
  
The QoS for the $i$th load class at time $t$ is a functional of the state of charge (SoC) $x_i(t)$ used in battery models for thermostatically controlled loads (TCLs)   \cite{haosanpoovin15} and residential pools    \cite{chebusmey14,meybarbusyueehr15}.    The SoC is assumed to evolve according to the linear system,
\begin{equation}
\ddt x_i(t)   = -\alpha_i x_i(t)   - z_i(t),    
\label{e:SoC_ODE}
\end{equation}
in which $-z_i(t)$ is power deviation at time $t$. 
Denoting the derivative by
\begin{equation}
u_i(t) \eqdef \ddt z_i(t),   
\label{e:ramp_nu}
\end{equation}
the \Cref{e:SoC_ODE,e:ramp_nu} constitute a linear dynamical system with state $(\bfmx, \bfmz)$ and control input $\bfnu$.

%\jjm{strongly convex instead of strictly convex}
For a given input $u$,  the resulting QoS for load class $i$ is quantified by the integral of cost:
\begin{equation}
\text{QoS}_i(u) =\int_0^{\clT}    c_i(x_i(t))  \, dt,   
\label{e:AggQoS}
\end{equation}
where each $c_i\colon\Re\to\Re_+$ is a strongly convex function. 
If this is small, then the aggregate SoC is small in an average sense.   This is a necessary condition for each of the loads in the $i$th class to remain within individual QoS bounds.   Simulation studies and analysis indicate that bounds on this aggregate QoS are also sufficient to ensure that a high percentage of the population will remain within individual target QoS levels,  subject to the homogeneity of the load collection, the control architecture,  the bounds on QoS,  and other factors \cite{chehasmatbusmey18, haosanpoovin15}.
\rwm{is "analys\underline{e}s" better than "analys\underline{i}s" here?}

As in   \cite{cammatkiebusmey18},  a finite-horizon optimal control problem over $[0,\clT]$ is considered,  with state cost given by 
\begin{equation}
\cX(x) =  \sum_{i=1}^{M}   c_i(x_i)  \,, \quad x\in\Re^M\,.
\label{e:cX}
\end{equation}  
The term \textit{state space collapse} comes from the literature on stochastic networks \cite{rei84b,CTCN}, which may be regarded as a special case of the model reduction obtained using singular perturbation methods \cite{sakorekok84}.   

The collapse demonstrated in this paper is obtained through the special structure of the dynamics of load and generation.
 The \textit{descriptor dynamics} associated with the $M$-dimensional  SoC process are obtained from    \eqref{e:SoC_ODE}:
\begin{equation}
  \ddt x_\sigma (t) = -\alpha^\transpose x(t)  -  z_\sigma(t), 
\label{e:descriptor}
\end{equation} 
in which  $x_\sigma(t)=\sum x_i(t)$,    $z_\sigma(t)=\sum z_i(t)$.   The evolution of  the aggregate power deviation is similar but simpler:
\[
\ddt z_\sigma (t) = u_\sigma(t).
\]
What is crucial here is that the individual inputs $\{u_i(t)\}$ are not subject to individual cost constraints: it is only the sum that is subject to a cost indirectly, through the cost on generation ramping.  Consequently, the optimal control formulations considered in this paper   fall in the category of ``cheap optimal control'' \cite{sakorekok84, fra79, hausil83}.    

For a cost that is quadratic in $(x,z_\sigma, u_\sigma)$, it follows from the main result of  \cite{fra79} that the solution to the infinite-horizon optimal control problem  has a simple form:   There is a one-dimensional subspace $\state^\star$ such that $x^\star(t)\in\state^\star$ for all $t > 0$. These results are extended to include the finite-time horizon optimal control problem in \cite{hausil83}; the optimal control evolves smoothly following a jump at time $t=0$ \cite[Theorem 5.8]{hausil83}. 

The findings in the non-quadratic, finite time-horizon setting of this paper are equally remarkable:  
the \textit{marginal costs} evolve in a two-dimensional subspace generated by a co-state process $\bflambda^\star$ and its derivative;  for $t>0$:
\begin{equation}
c'_i\,   (x^\star(t)) = \alpha_i \lambda^\star(t) - \ddt \lambda^\star(t).   
\label{e:MC}
\end{equation} 
Consequently, optimal SoC for each of the $M$ load classes evolves in a two-dimensional manifold. 
These   results are surveyed in \Section{s:collapse}.

% Given the status of marginal cost in economics, it isn't surprising that there are economic consequences.
	
\subsection{Implications for markets}

\Section{s:idiots} reviews dynamic competitive equilibrium theory, and explains that the negative of the co-state function may be interpreted as a price in a competitive equilibrium.   Extensions of results from  \cite{chomey10, wankownegshamey10} are obtained, relating average marginal cost of generation and average marginal value of load classes to average prices.   An example of the competitive equilibrium price for a model in which the QoS cost functions are all quadratic is shown in 	\Cref{f:idiotsprice}.    In this example, the net-load is piecewise constant,  and the resulting optimal generation is relatively smooth.  The price is also smooth, and anticipates the surge in load well before its occurrence.   The details of the simulation are contained in \Cref{s:numerics}.
\rwm{are fWH, sWH, Fr and PP in Fig 1 defined at this point?}

The results in \Section{s:idiots} close a logical gap in the prior work  \cite{chomey07,chomey10, wankownegshamey10,kizman10b,zavani11a} (and many others),  in which  the utility function for the consumer is assumed to be a concave function of power, rather than a natural metric such as the QoS  \eqref{e:AggQoS}.   
\spm{Sep 8:   I should find a paper by Wolak or Hobbs or Hogan.   Wolak always tells me that he was way ahead of me [he had established the conclusions of chomey10 years prior], but he never shared any evidence.   
}

Consumer QoS is included in the formulation of the competitive equilibrium analysis in \cite{lichelow11},  but in this prior work, it is assumed that power consumption from each residential load can be varied continuously; however, the residential loads considered for demand-side management are typically ON/OFF devices. In the present paper, resource allocation is performed over load aggregations modeled as virtual batteries.  Consequently, the SoC and power trajectories can be assumed to be smooth---see \cite{matbusmey17, meybarbusyueehr15,chebusmey17a}, which discuss the mean-field characteristics of load aggregations under demand dispatch.

\begin{figure}
	\centering
	\vspace{0.5em}
\includegraphics[width= .8\hsize]{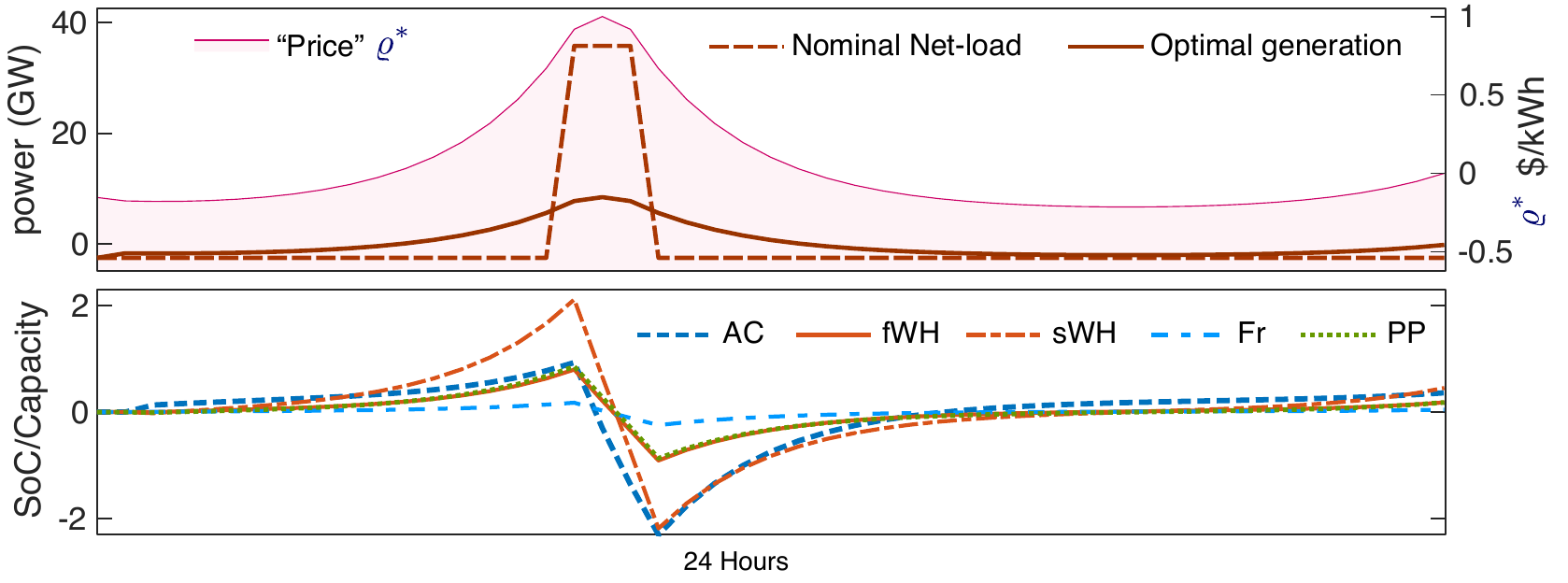}
	%\Ebox{1}{}
	\vspace{-1.25em}
	\caption{The competitive equilibrium price $\rho^\star=-\lambda^\star$ when the nominal net-load $\ell$ is piecewise constant. 
	 The quantity and distribution of loads were chosen based on California's usage  \cite{cammatkiebusmey18,matdyscal15}. 
	 The residential air conditioners and commercial water heaters exceeded capacity bounds (by a few degrees for a few hours). 
 The optimal generation $g^\star$ is nearly constant, despite the 40GW load surge.  }
	\label{f:idiotsprice} 
		\vspace{-1.25em}
\end{figure}

\textit{These economic findings should be viewed with caution.}  It would be naive to think that real-time control can be achieved using the price signals discovered in \Section{s:idiots}.     It is typically assumed that price discovery occurs through an auction.  Do we believe these non-causal prices will emerge from an auction?    
In the experiments surveyed in \Cref{s:numerics}, the cost functions are modified so that the consumers experience no loss of QoS, and  the resulting average price is nearly zero.  How then can these prices provide any incentive for participation?

The resource allocation problem introduced in \cite{cammatkiebusmey18} and investigated in this paper is intended to be part of a model predictive control (MPC) architecture,  while control through price signals  is inherently open loop.    Imagine attempting to apply MPC using price signals?   The aggregators managing the loads would not be able to predict re-calculated prices and might make decisions that would create personal losses or windfalls purely based on the control architecture.    This would create uncertainty in the control solution and the long-term sustainability of the grid-aggregator relationship.

For these reasons, it is assumed here that a balancing authority solves the optimal resource allocation problem,  and   control of individual loads is achieved through automation. 
%The foundations of price-based demand response  are rooted in the theory developed by Dupuit in the 1840s, Hotelling in the 1930s, Vickrey in the 1950s, and Schweppe in the 1980s \cite{dupuit1844,hotelling1937,vickrey1955,schweppe1988}, which look at the relationship between the marginal cost to provide services, such as railways, water systems, and electricity, and the concomitant value to the consumers.  However,  in contrast to the direct value to the consumer of water delivery or transport, electricity has always presented a challenge because the product that the customer values is not electricity per se, but the heating, cooling, or lighting that comes from it. \textit{This disconnect fundamentally alters the prices-to-loads compact.}

There is history that supports the belief that automation rather than price signals is the most efficient and reliable way to control the grid.
Following deregulation of the industry in the 1980s, telephone companies investigated real-time pricing (alternately called ``measured rates'') for local service, based on the assumption that measured-rate pricing could substantially increase economic efficiency.  However, looking specifically at the \textit{net welfare effects} on customers, studies found that measured rates often result in efficiency losses. In addition, the telephone subscribers had difficulty understanding and coping with complicated tariff structures, and the systems required to implement them were too complex \cite{park1987}.  With the integration of distributed energy resources, which are characterized by high fixed costs and essentially zero variable costs, the grid increasingly has characteristics similar to the telecommunications industry. Further discussion on these issues is contained in \cite{lobluhinmey19}.

%These issues are particularly prescient in dynamic price-based demand response. Studies have found that dynamic price-based demand response has a disproportionately negative impact on indigent, disabled, and elderly consumers, as \textit{critical peak prices} and \textit{time-of-use rates} can be a punitive signal to these consumers to reduce their power usage, even at a detriment to their quality of life \cite{ale10}.

\smallskip

The remainder of this paper   is organized as follows.  The dynamic control problem is introduced in  \Cref{s:control}, which is based on the model introduced in   \cite{cammatkiebusmey18}.   The major departure from \cite{cammatkiebusmey18} is the relaxation of hard constraints on any resource.   This is reasonable for control outside of a major crisis. The preliminaries required for analysis of the optimal solution are introduced in \Cref{s:prelim}. The main results surveyed in \Cref{s:collapse} demonstrate state space collapse, and economic implications are contained in
\Cref{s:idiots}.   Numerical simulations are provided in \Cref{s:numerics}.   Conclusions and directions for future research are discussed in \Cref{s:con}.

%As other research has demonstrated \cite{vardakas2015}, the industry has approached \textit{demand side management} from several angles.  Some have tested time-of-use rates like tried in the telecommunications industry.  Some have proposed using incentives to encourage more efficient use, but without any penalty or increased cost even during peak system conditions.  A third approach considered assumed direct load control, either on a centralized/aggregated basis or distributed (i.e., direct control of devices). And while \cite{vardakas2015} doesn't draw any conclusions regarding the varying approaches, it does show the extent of the work that has been pursued.

\section{Preliminaries}
\label{s:prelim}

\subsection*{Notation}  
\noindent
$\clT$ : time horizon for control\\
$\ell(t)$ : net-load on $[0,\clT]$,   and 
$\barell= \clT^{-1}\int_0^{\clT}\ell(t) \, dt$  
\\
$g(t)$ : power from traditional generation; 
$\gamma(t) = \ddt g(t)$.
\\
$M$: number of load classes, indexed by $i \in \{1,...,M\}$.
\\
$x_i (t)$ : state of charge (SoC) of load class $i$ \\
$-z_i(t)$ : power deviation from load class $i$; 
$u_i(t)  = \ddt z_i(t)$.
\\
Subscript ``$\sigma$'' denotes sum, e.g.,
$x_\sigma(t) = \sum_i x_i(t) $ 
\\
%When there is risk of ambiguity, bold face is used to distinguish functions, such as $\bfmg = \{g(t) : 0\le t\le\clT\}$. In addition, the vector notation is used to avoid ambiguity when introducing the co-state vectors in the Hamiltonian (e.g. $\lambda \in \mathbb{R}^M$).
For a function $K\colon \Re^3\to\Re $ we adopt
%$c_i$ is the cost on SoC of VES class $i$.\\
%$\cG$ is the cost on generation.\\
%$\cdG$ is the ramping cost on generation. \\
%By a VES or load class, we imply the aggregate virtual energy storage provided by a homogeneous collection of flexible loads. \\
%In cases of ambiguity, bold variables denote functions of time. \\
%$f'$ represents the derivative of a function $f$. \\
  the standard calculus notation: for $\xi \in \mathbb{Z}_+^3$,  
$$
K_\xi (r, s, t) \eqdef \bigl(\FRAC{1}{\partial }{\partial r}\bigr)^{\xi_1} \bigl(\FRAC{1}{\partial}{\partial s}\bigr)^{\xi_2} \bigl(\FRAC{1}{\partial}{\partial t}\bigr)^{\xi_3} K(r, s, t) 
$$ 

\iffalse
$\mathbf{0} \text{ and } \mathbf{1}$ denote the $M$-dimensional column vectors of $0\text{'s and }1\text{'s}$, respectively, and depending on the context, $\mathbf{0}$ may denote an $M\times M$ matrix of all $0$'s. $I$ is the identity matrix.

The linear dynamical system represented by \Cref{e:SoC_ODE,e:ramp_nu} is alternately represented as follows:
\begin{equation}
\begin{bmatrix}
\dot x(t) \\
\dot z(t)
\end{bmatrix}
= A 
\begin{bmatrix}
x(t) \\
z(t)
\end{bmatrix}
+ B u(t)
\label{e:stateMatrix}
\end{equation}
where the matrices
\[
A = 
\begin{bmatrix}
-\text{diag}(\alpha) & -I \\
\bfmath{0} & \bfmath{0}
\end{bmatrix}
\qquad B = 
\begin{bmatrix}
\bfmath{0}\\
I
\end{bmatrix}
\]
\fi

\subsection{Resource allocation control architecture}
\label{s:control}

\iffalse
\begin{figure*}
	\Ebox{.75}{GridArchitectureCDC2018_lr.pdf}
	\vspace{-.5em}
	\caption{The control architecture: feedforward commands are computed from forecasts and feedback commands are computed from real-time error.}
	\label{f:GridArchitectureCDC2018} 
		\vspace{-.5em}
\end{figure*}
\fi

The optimal control architecture is defined by a convex program   over the time-horizon   $[0,\clT]$: 
\begin{mini!}
	{g, \gamma, x }{   
	\int_0^{\clT}    \big[   c_g(g(t) ) + c_d(\gamma(t))  + \cX(x(t)) \bigr] \, dt }{}{}
	{\label{qp19}}{}
	\addConstraint{\ell(t)}{=g(t)+z_\sigma (t)}
	{\label{e:balancecons}}{}
	\addConstraint{\ddt g(t)}{=\gamma (t)}
	{\label{e:genrampcons}}{}
	\addConstraint{\ddt {x}_i(t)}{=- \alpha_i x_i(t) - z_i(t)}
	{\label{e:soccons}}{}
	\addConstraint{\ddt z_i(t)}{= u_i(t), \quad i\ \in \{1, ..., M\}}
	{\label{e:loadrampcons}}{}
%	\addConstraint{\text{All constraints on $\bfmx_i, \bfmz_i,  \clL_i$}}
%	\addConstraint{-C_i}{\leq x_i(t)} {\leq C_i}
%	\addConstraint{-\eta_{i-}}{\leq z_i(t)} {\leq \eta_{i+}}
%	\addConstraint{\dot\clL_i(t)}{=- \alpha_i \clL_i(t) - z_i(t)}
%	\addConstraint{b_{i-}}{\leq \clL_i(t)} {\leq b_{i+}}
%	\addConstraint{-\eta_{i-}}{\leq z_i(t)} {\leq \eta_{i+}}
	%\addConstraint{0}{=\int_{t=0}^{\clT}z_i(\tau)d\tau}{} 
\end{mini!}
with ${x(0), z(0) \in \Re ^M \text{ given}}$.  These equations are adapted from eq.~(14) of  \cite{cammatkiebusmey18}, but the motivations and assumptions here are different.

The analysis here allows general strongly convex and twice continuously differentiable cost on SoC and generation, but maintains the quadratic cost on ramping imposed in \cite{cammatkiebusmey18}:

\As{A1}
The net load $ {\ell}$ is $C^1$.		
The cost functions $\{c_i \} \text{ and } c_g$ are non-negative, class $C^2$,  and   strongly convex: $c_i''(x)  \ge \mu$ for some $\mu>0$ and all $i, x$.  The ramping cost is quadratic:  for fixed $\kappa>0$,
\begin{equation}
c_d(x) = \kappa x^2   ,\qquad x\in \Re.
\label{e:genCost}
\end{equation}
\qedIEEE

The objective is to minimize the cost of traditional generation, the cost on ramping of traditional generation, and the cost on SoC of load classes  subject to: (i) the constraint of balancing the net load $\bfmell$ with traditional generation $\bfmg$ and power deviation from flexible loads $\bfmz$, \eqref{e:balancecons}; (ii) the dynamics of generator ramping, \eqref{e:genrampcons}; and (iii) the dynamical constraints on load classes, \eqref{e:soccons} and \eqref{e:loadrampcons}.

% \rd{jm: do we need statement or proposition on existence of equilibrum point for the convex program? Lookup Luenberger}
%
%The main differences in the assumptions: 
%\begin{romannum}
%\item In  \cite{cammatkiebusmey18} it is assumed that all of the costs are quadratic.

% 
%
%
%\item   Ramping costs  modeled in      \cite{cammatkiebusmey18}  as a quadratic cost on the derivatives of $\{ x_i(t)\}$ are relaxed in the present work, as their inclusion was not well-motivated.  
%
%\item  Hard constraints on $\{x(t), z(t)\}$ are relaxed.
%\end{romannum}
 
% The inclusion of ramping costs for a class of loads was not well-motivated.  It was included in  \cite{cammatkiebusmey18}   to obtain a strongly convex objective function to ease computation.   

 % The relaxation of hard constraints is not very restrictive --  \bl{ [explain at top that hitting the boundaries is nuts -- realization results in infinite cycling]  [don't be too aggressive with numerics -- we are not trying to reproduce the results from \cite{cammatkiebusmey18} where the hard constraints on $z(t)$ really were important.  ]}

% \rd{JM: moved following essay to beginning of next section}

%\jjm{reintroduce $c_g$ instead of quadratic cost on generation}

The optimization problem
\eqref{qp19}  may be regarded as a fully observed, finite-horizon control problem. In order to put this in state space form, it is necessary to eliminate the algebraic constraint \eqref{e:balancecons}. The resulting state process is $[x(t), z(t)] \in\Re^{2M}$, the input is  $u(t) \in\Re^{M}$, and the cost function is obtained by eliminating $g(t)$ and $\gamma(t)$ from \eqref{qp19} using \cref{e:balancecons,e:genrampcons,e:genCost} to obtain
\begin{align}
c(x(t), z(t), u(t), t) & =   \cX(x(t)) + c_g(\ell(t) - z_\sigma(t))  +           c_{\tilde d}(u(t),t),     
\label{e:modCost} 
\\[.5em]
\text{
where}\qquad 
  c_{\tilde d}(u(t),t)   &=  \kappa \bigl[   u_\sigma(t) - \ddt\ell(t) \bigr]^2 
\label{e:singCost}
\end{align}
The total cost in  \eqref{qp19} is the integral of  \eqref{e:modCost}.  This is a singular control problem because the control cost is degenerate:  the terms involving the control cost in \eqref{e:modCost} are expressed purely in terms of the sum $u_\sigma(t)$.        It is found that this singularity   is a great benefit for obtaining   structure for the optimal control solution.

\subsection{Value functions}

For $t_0 \in [0, \clT)$,     the cost-to-go is denoted,
\begin{align}
\label{e:V1}
\begin{split}
& J^\star(x,z,t_0) \eqdef  \displaystyle \inf_{u_{[t_0, \clT]}} \int_{t_0}^{\clT}     c(x(t), z(t), u(t),t) \, dt, \\
\end{split}
\end{align}
where the infimum is over continuous $u$,   subject to
\eqref{e:balancecons}--\eqref{e:loadrampcons},   and   with $x(t_0) = x, z(t_0) = z$ given.  

\Cref{p:JeqK} asserts that the cost-to-go can be expressed purely as a function of $(x_\sigma, z_\sigma)$. This is the first evidence of state space collapse.
Denote
\begin{align}
\label{e:V2}
\begin{split}
& K^\star(x_\sigma,z_\sigma,t_0) \eqdef  \displaystyle \inf_{x^+, z^+} J^\star(x,z,t_0), \\
\end{split}
\end{align}
where the infimum is over $x^+,z^+\in\Re^M$, subject to the linear constraints $ x_\sigma^+ = x_\sigma$,  $ z_\sigma^+ = z_\sigma$.

\begin{proposition}
\label[proposition]{p:JeqK}
	The following hold under Assumption~(A1): for each $t_0 \in [0, \clT)$,
\begin{romannum}
		\item $J^\star$ is convex in $(x,z)$ and finite-valued.
		\item $K^\star(x_\sigma,z_\sigma,t_0) = J^\star(x,z,t_0)$ for each $x, z \in \Re^M$  
		\qedIEEE
	\end{romannum}
\end{proposition}

The proof of 	 \Cref{p:JeqK} and most of the  results that follow are contained in the appendix.

For a given initial condition $x,z$,  with the optimal state trajectory  $\{x^\star(t), z^\star(t)  :  0< t\le \clT \}$,   denote
\begin{align}
\begin{split}
\lambda^\star(t) & = K^\star_{1,0,0} (x_\sigma^\star(t),z_\sigma^\star(t) ,t), \\
\beta^\star(t) & = K^\star_{0,1,0} (x_\sigma^\star(t),z_\sigma^\star(t) ,t)\,,  \qquad  0< t \le \clT.
\end{split}
\label{e:lambda}
\end{align}

In addition to (A1), the following assumptions are imposed throughout the remainder of the paper:   
 \jjm{Sept 7: new text that could be added: In particular, the existence and uniqueness of the optimal solution can be shown using an extension to Filippov's theorem \cite{fil62} and the fact that the cost functional is convex, but it is assumed (A2) in this paper.
 \\
 Joe, tell us your thoughts!}
%  In \eqref{e:V2}, $\Bigl(x^+, u_{[t_0, \clT]}\Bigr)$ are treated as independent variables for inifimization.

\As{A2}
	For each $t_0\in [0,\clT)$ and each initial condition $(x,z)$,    the optimal control problem admits a unique solution $\{x^\star(t), z^\star(t), u^\star(t) : t_0< t\le \clT \}$  satisfying   
\begin{romannum}
		\item[(a)]  
		$(x^\star(t), z^\star(t) )$ is $C^1$ on the semi-open interval $  (0, \clT]$. 
		\item[(b)]
		There are right hand limits at  $ t_0$, denoted
		\begin{equation}
		x^\star(t_0) = \lim_{t\downarrow t_0}  x^\star(t)\,, \qquad
		z^\star(t_0) = \lim_{t\downarrow t_0}  z^\star(t),  
		\label{e:state-star}
\end{equation}
 satisfying $x^\star_\sigma(t_0) = x_\sigma $,   $z^\star_\sigma(t_0) =z_\sigma $. 
	\end{romannum}
	
%	\spm{Joe thinks we should eliminate A2 since it follows from A3.   I agree that (a) does, but we still need (b).   Leaving it as is for now}

\As{A3}
		The value function $K^\star\colon\Re\times\Re\times[0,\clT]\to\Re$ is $C^1$.

\As{A4}
		 The function $\bflambda^\star$ is $C^2$.
		 
\qedIEEE

\iffalse
\spm{I don't think this is needed:
The value of $(x^\star,z^\star)$ in \eqref{e:state-star}
depends of course on $(x,z,t_0)$, but this dependency will be supressed for ease of notation.
}

\spm{Distracting and not clear enough:   (the first statement might be used, but modified to explain what we mean by relaxed)
\\
These assumptions will be relaxed in an extended version of this paper (in preparation).  
\\
 It is convenient to first impose them so that structure of the optimal solution is easily revealed.    We then abandon the assumption through the following arguments:  first, it is shown that the  coupled differential equations  \eqref{e:soccons}--\eqref{e:loadrampcons} admit a solution.  The proof of existence then follows on establishing that the solution defines an optimal trajectory that satisfies assumptions (A1) to (A4).
 }
\fi

\section{State space collapse}
\label{s:collapse}

\Cref{t:main}
 unveils the structure of the optimal solution:  in particular, the $M\text{-dimensional}$ optimal state process $\bfmx^\star$ evolves on a two-dimensional manifold.
 
\begin{theorem}
	\label{t:main}
	For $t \in (0,\clT]$, the optimal solution $(x^\star(t), z^\star(t), u^\star(t), \lambda^\star(t), \beta^\star(t)) \in \Re ^{3M+2}$ is the solution to the following system of $3M + 2$ equations:
	\begin{subequations}
		\begin{align}
		\label{e:SoC}
		\ddt x^\star_i(t)   &= -\alpha_i x^\star_i(t)   - z^\star_i(t),\\
		\label{e:contro}
		\ddt z^\star_i(t) &= u^\star_i(t), \\
		\label{e:2dsub}
		\ddt \lambda^\star(t) &= - c'_i\,   (x^\star_i(t)) + \alpha_i \lambda^\star(t), \\
		\label{e:betamain}
		\ddt \beta^\star(t) &=  c'_g(\ell(t) - z^\star_\sigma(t)) + \lambda^\star(t), \\
		\label{e:rampsum}
		u^\star_\sigma(t) &= \ddt \ell(t) - \frac{1}{2 \kappa} \beta^\star(t),
		\end{align}
	\end{subequations}
	where $i \in \{1,...,M\}$, with the boundary conditions $x(0+), z(0+), \lambda^\star(\clT) = 0, \text { and }  \beta^\star(\clT) = 0$.
	\qedIEEE
\end{theorem}

Equation \eqref{e:2dsub} has a remarkable interpretation: the marginal costs for the $M$ different load classes evolve in a two-dimensional subspace generated by the functions $\{\lambda^\star(t),\ddt \lambda^\star(t) :   t\in (0, \clT] \}$. Since $c_i$ is strictly convex, $c'_i$ is strictly monotone and invertible. Consequently, the optimal SoC evolves on a two-dimensional manifold:
% generated by $\{\lambda^\star(t),\ddt \lambda^\star(t) :   t\in (0, \clT] \}$:
\begin{equation}
\label{e:xstar}
x_i^\star (t) =  (c'_i)^{-1} (\alpha_i  \lambda^\star(t) - \ddt \lambda^\star(t)).
\end{equation}

%\begin{remark}[Optimal mapping at initial time]
%	\label{r:IC}
	It follows from \eqref{e:lambda} and  Assumption (A2b) that $\lambda^\star(0+) = K^\star_{1,0,0} (x_\sigma^\star(0),z_\sigma^\star(0) ,0)$. Then, in consequence of  \eqref{e:2dsub} and \eqref{e:SoC}, we have the following:
	\begin{corollary}
The optimal mapping $ (x^\star(0+), z^\star(0+) )$   is obtained by
\begin{equation}
\begin{aligned} 
\! \! \! 	\! \! \! \! \!
&
c_i'  (x_i^\star(0+)) =  \alpha_i K^\star_{1,0,0}(x^\star_\sigma(0), z^\star_\sigma(0),0) \\
& \qquad \qquad -	\ddt K^\star_{1,0,0}(x^\star_\sigma(t), z^\star_\sigma(t),t) \Big|_{t=0}	
\\
&
z_i^\star(0+)  = - \alpha_i x_i^\star(0+)  -\frac{1}{ c_i'' (x_i^\star(0+)) } \bigl[ \alpha_i 	\ddt K^\star_{1,0,0}(x^\star_\sigma(t), z^\star_\sigma(t),t)\big|_{t=0}   - 	\dddt K^\star_{1,0,0}(x^\star_\sigma(t), z^\star_\sigma(t),t)\big|_{t=0} \bigr]
\end{aligned} 
\label{e:IC}
\end{equation}
\spm{Warning on Aug 16:  I fixed significant errors in the first equation!  We need to work on this equation to see if we really have a mapping. \rd{JM: I think it is correct now.}}
%\end{remark}
\qedIEEE

	\end{corollary}

The remainder of this section concerns two very different interpretations of $\bflambda^\star$.   The first is easily predicted.

%The inverse exists because $c'_i$ is strictly increasing and continuous.

\subsection{$\bflambda^\star$ as co-state}

The Hamiltonian with co-state variables $\lambda, \beta \in \Re ^M$ corresponding to system equations \eqref{e:soccons} and \eqref{e:loadrampcons}, respectively, is denoted:
\begin{align}
\label{e:H}
\begin{split}
\!\!\!\!
H (x, z, u, {\lambda}, {\beta}, t) & \eqdef c (x, z ,u, t )  
\\
& \quad + \sum_i \lambda_i (- \alpha_i x_i - z_i) + \sum_i  \beta_i u_i
\end{split}
\end{align}
This notation  for the co-state variables would appear to conflict with the notation in 	\Cref{t:main}.   The choice of notation  is made clear in the following:

\begin{proposition}
	\label[proposition]{r:lambda}
	Associated with the optimal input-state  $(\bfnu^\star, \bfmx^\star, \bfmz^\star)$ are a pair of co-state variables $\bflambda^\star, \bfbeta^\star $ evolving in $ \Re ^M$ and satisfying for $ 0<t\le \clT$,
	\[
	(x^\star(t), z^\star(t)) =\argmin_{(x,z)} 
	H (x, z, u^\star(t), \lambda^\star(t), \beta^\star(t), t) 
	\]
	
	For each  $i \in \{1,...,M\}$ and $t \in (0, \clT]$,
	\begin{align}
	\label{e:lambdas}
	\begin{split}
	\lambda^\star_i(t) &=  K_{1,0,0}^\star(x^\star_\sigma(t), z^\star_\sigma(t), t) = \lambda^\star(t),
	\\
	\beta^\star_i(t) &= K_{0,1,0}^\star(x^\star_\sigma(t), z^\star_\sigma(t), t) = \beta^\star(t).  
	\end{split}
	\end{align}
	\qedIEEE
\end{proposition}

\subsection{$\bflambda^\star$ and a Lagrangian decomposition}

Rather than eliminate the variable $\bfmg$ using \eqref{e:balancecons},   new insight is obtained on maintaining $\bfmg,\bfgamma, \bfmx$
as variables in the optimization problem. First, construct a Lagrangian relaxation with Lagrange multiplier $\bfvarrho $, as follows:     
\begin{gather}
\label{e:phistar}
\begin{split}
\phi^\star(\bfvarrho) &=\!\! \inf_{g, \gamma, x} \int_0^\clT \Bigl \{ c_g(g(t) )   
+ c_d(\gamma(t)) + \cX(x(t))  \\ 
&   \quad\qquad   + \varrho(t) (\ell(t) - g(t) - z_\sigma(t)) \Bigr\} \,dt,
\end{split}
\end{gather}
where as usual, the infimum is  subject to \eqref{e:genrampcons}--\eqref{e:loadrampcons},  with given initial conditions.  % $x^+, z^+$ given.

This amounts to a Lagrangian decomposition, consisting of the following $M+1$ independent optimization problems:
\begin{romannum}
	\item Minimization problem over $\{ g(t), \dot g(t) \}$:
	\begin{equation}
	\label{e:grid}
	\inf_{g} \int_0^\clT \clL_g(g(t), \dot g(t), t) \,dt,
	\end{equation}
where, based on the definition \eqref{e:genrampcons},
$$
	\clL_g(g(t) , \dot g(t), t)  =  c_g(g(t) )   
	+ c_d(\dot g(t)) -\varrho(t)  (g(t) -\ell(t)).
$$
	 
 	\item 
	Minimization problem over $\{ x_i(t) ,\dot x_i(t) \}$:	
\begin{equation}
	\inf_{x_i} \int_0^\clT \clL_i(x_i(t), \dot x_i(t), t) \,dt,
	\label{e:aggregator}
	\end{equation}
	where, after accounting for the constraints  \eqref{e:soccons}, \eqref{e:loadrampcons},
 $$
 \clL_i(x_i(t), \dot x_i(t), t)  = c_i\bigl(x_i(t) \bigr) + \alpha_i\varrho(t)  x_i(t) + \varrho(t) \dot x_i(t).
 $$ 
\end{romannum}

The Euler-Lagrange equations lead to equations for the optimizers:
%\spm{Check my change:   I wrote that the equations hold provided a smooth opt exists.    I think this is all we need, but we need to discuss}
\begin{proposition}
\label[proposition]{p:conrho}
For any function $\bfvarrho$ that is continuously differentiable on  $(0,\clT]$,	if  $\bfmg^\varrho$ and $\bfmx_i^\varrho$ are $C^1$ optimizers for the minimization problems in   \cref{e:grid,e:aggregator},  then they solve the following  differential equations:  \spm{Sep 8:  not sure I'd call these "ordinary" diff eqns.   It would be more accurate to say that we have a mapping from $(x,g)$ to $\rho$ that is defined by an ODE.  }
\begin{align}
	\label{e:ELg}
	\cG'\,   (g^\varrho(t))   - \ddt \cdG'(\dot g^\varrho(t)) & = \varrho(t), \\
	\label{e:EL}
	c'_i\,   (x_i^\varrho(t)) + \alpha_i \varrho(t) - \ddt \varrho(t) & = 0,
	\end{align}
	with  boundary conditions $\cdG'(\dot g^\varrho(\clT)) = \varrho(\clT) = 0$.
	\qedIEEE
	\end {proposition}
	\bigskip

The dual functional $\phi^\star$ satisfies   weak duality: $ \phi^\star(\bfvarrho) \le J^\star(x(0) , z(0), 0)$ for any  $\bfvarrho$,   and the dual convex program is defined as $\sup_\rho  \phi^\star(\bfvarrho) $.     
 The solution to the dual is obtained by   combining \Cref{p:conrho} and \Cref{t:main},  and from this we obtain strong duality:

\begin{proposition}
\label[proposition]{p:rhoAndLambda}
The dual admits an optimizer given by 	 
	\[
	\varrho^\star(t) = -\lambda^\star(t)  \,, \qquad  t \in (0, \clT].
	\]
	\qedIEEE
\end{proposition}

\iffalse
\subsection{Example: Linear Quadratic Cost}

\jjm{Perhaps we can state the Riccati equation without proof in CDC. Riccati write up goes here.}
\fi

\section{Real time prices}
\label{s:idiots}

A dynamic competitive equilibrium model is introduced in this section.   \Cref{p:rhoAndLambda} is applied to establish the existence of a competitive equilibrium,  and other results from the previous section provide approximations of the average price in terms of both the  average marginal value and the average marginal cost.

 We begin by recalling basic concepts.

% relates the optimal Lagrange parameters for the VES dynamics obtained in the preceding section to the solution of a dynamic competitive equilibrium (CE) problem. We begin by explaining these concepts. 

A  ``snapshot'' commodity market for a divisible good $G \in \Re $ is defined based on two ``utility functions'':     for consumption, $\clU_D \colon  \Re  \to   \Re $,   
   and for supply,  $\clU_S:  \Re  \to   \Re $.    
The \textit{social planner's problem} (SPP) of macro-economics is defined as the optimization problem,
\begin{equation*}
\max_{G} ~ \clU_D(G) + \clU_S(G).
\end{equation*}
A solution is called an \textit{efficient allocation}.

If the utility functions are strictly concave, and there exists an optimizer $G^\star$, then there is a unique price $\varrho^\star$ that achieves the so-called \textit{competitive equilibrium} (CE):
\[
G^\star = \underset{G} {\arg\,\max} ~ \clU_D(G) - \varrho^\star G =  \underset{G} {\arg\,\max} ~ \clU_S(G) + \varrho^\star G.
\]
The price $\varrho^\star$ is the Lagrange multiplier associated with $G_D = G_S$ in the equivalent formulation of the SPP \cite{tak85}:
\begin{equation}
\begin{split}
&\max_{G} ~ \clU_D(G_D) + \clU_S(G_S), \\
&\text{ s.t. } {G_D = G_S}.
\end{split}
\label{e:CEsnap}
\end{equation}

Formulations of dynamic CE theory address problems in which $G$ is a function of time and subject to various constraints 
\cite{tak85}; see \cite{chomey07,chomey10, wankownegshamey10,kizman10b,zavani11a} for theory in   the context of power systems.

The dynamic CE model considered in this section involves $M+1$ players:   there is a single supplier (or class of suppliers) that provide traditional generation $\bfmg$  and M consumers with power deviation $-\bfmz_i$,  for $1\le i\le M$.    The utility function for each player is the negative of the cost on the SoC:
$\clU_{D_i}(z_i) =   - c_i\bigl(x_i)$ and $\clU_S(g,\dot g) = -c_g(g)  -c_d(-\dot g)$.  The SPP is defined as follows:
\begin{gather}
\label{e:CEdyn}
\max_{g,z_i}  
 \int_{0}^{\clT}  \Bigl\{  
 		 \clU_S(g(t) ,\dot g(t))  +  \sum_{i=1}^M  \clU_{D_i}(z_i(t))   
		 		\Bigr\}  \, dt ,
\end{gather}
subject to the balancing and dynamic constraints  \eqref{e:balancecons}--\eqref{e:loadrampcons}. This is equivalent to the optimization problem \eqref{qp19} analyzed in the previous sections.

The Lagrangian decomposition behind   \Cref{p:rhoAndLambda} is analogous to the Lagrangian relaxation of \eqref{e:balancecons} in \eqref{e:CEdyn},  and hence,   $\varrho^\star(t) = -\lambda^\star(t)$ is the  competitive equilibrium price at time $t$, for $0<t\le\clT$.  That is, the optimizers   of \Cref{e:grid,e:aggregator} correspond to:
\[
\begin{aligned}
 \bfmz_i^\star      & =  \argmax_{z_i}     \int_0^\clT    \clU_{D_i}(z_i(t))    +  \varrho^\star(t) z_i(t) \, dt ,
\\
\bfmg^\star &= \argmax_g   \int_0^\clT  \clU_S(g(t) ,\dot g(t))    +\varrho^\star(t)   g(t)   \, dt  .
\end{aligned} 
\]
	 
Two quantities of special interest in  CE theory are the marginal cost and the marginal value at equilibrium.     
The marginal cost at time $t$ is  $ \cG'(g^\star(t))$,  and the marginal value is defined as the negative of marginal cost for each load:   $- c'_i\,   (x_i^\star(t))$. The averages of these quantities and the equilibrium price are denoted, respectively,  
\begin{align*}
\varrhoav & =  \frac{1} {\clT} \int_0^\clT \varrho^\star(t)\,dt, \\
\MCav & =  \frac{1} {\clT} \int_0^\clT  \cG'(g^\star(t))  \,dt, \\
\MVav & =  -\frac{1} {\clT} \int_0^\clT   c'_i\,   (x_i^\star(t))\,dt .
\end{align*}

In the snapshot CE model, it is known that the price coincides with both marginal value and marginal cost --- this is immediate from the Lagrangian decomposition of \eqref{e:CEsnap}.    This conclusion fails in general in a dynamic setting.   For the dynamic CE models considered in \cite{chomey10, wankownegshamey10}, it is shown that the price coincides with the marginal value (in which the marginal value is   defined with respect to power consumption), and the \textit{average price} is approximated by  average marginal cost.   These conclusions admit the following extension to the CE model introduced in this paper:   
\begin{proposition}
\label[proposition]{p:avgprice}
The average of the competitive equilibrium price admits the following approximations:
\begin{romannum}
\item    \textit{Weighted average marginal value}, plus $O(1/\clT)$:
	\begin{equation}
	\varrhoav = \frac{1}{\alpha_i}\MVav +  e_i^d/\clT\,,\qquad 1\le  i\le M.
	\label{e:MVavg} 
	\end{equation} 
\item    \textit{Average  marginal cost}, plus $O(1/\clT)$:
	\begin{equation}
	\varrhoav = \MCav + e^g/\clT,
	\label{e:MCavg} 
	\end{equation}
\end{romannum}
where the error terms are the differences,
\[
		e_i^d =    [\varrho^\star(\clT) - \varrho^\star(0) ]  /{\alpha_i  }
				\,,\quad
		e^g =\cdG'(\dot g^\varrho(0))    -    \cdG'(\dot g^\varrho(\clT ))    .
	\]
\qedIEEE
\end{proposition}

\iffalse
Given the periodic nature of the net load over $\clT = 24$ hours, it is reasonable to assume that $g(0) \approx g(\clT)$ and $\varrho^\star(0) \approx \varrho^\star(\clT)$. Taking the averages of the equations \eqref{e:ELg} and \eqref{e:EL} leads to the following remarkable result: at equilibrum, the daily average price is approximately equal to  the average marginal value of each load class $i$ scaled by $\alpha_i$ and the average marginal cost of generation.
\begin{align}
\begin{split}
 \frac{1} {\clT} \int_0^\clT \varrho^\star(t)\,dt & \approx \frac{1}{\alpha_i \clT} \int_0^\clT -c'_i\,   (x_i^\star(t))\,dt \\
 & \approx \frac{1} {\clT} \int_0^\clT  \cG'(g^\star(t))\,dt
 \end{split}
\end{align}
\fi

\begin{figure}[h]
	\vspace{0.4em}
	\Ebox{.75}{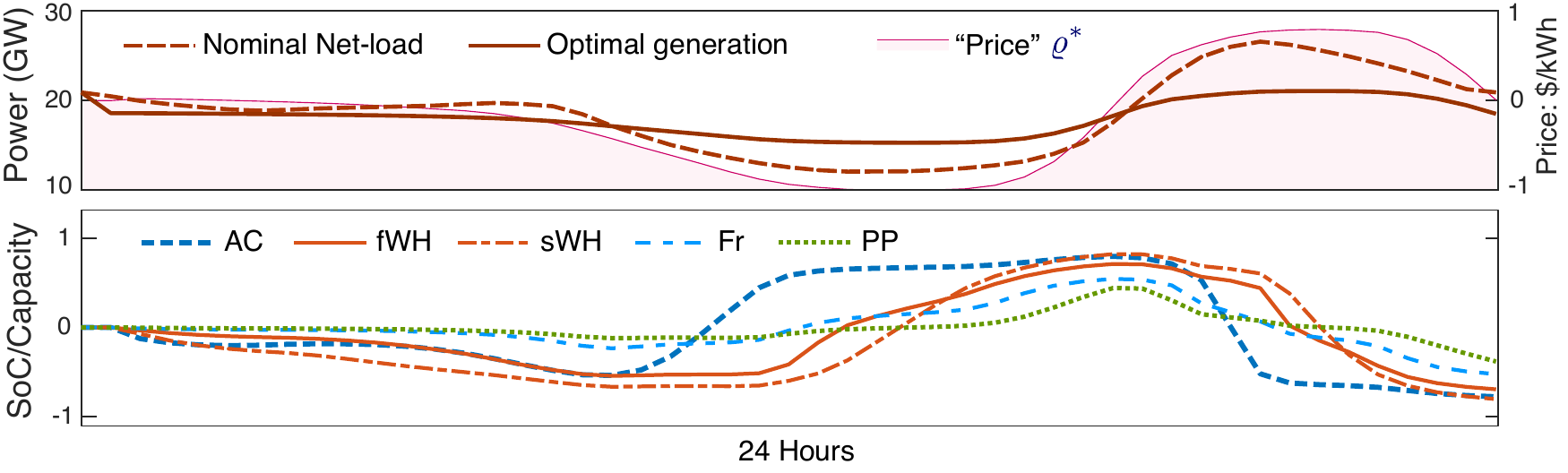}
	\vspace{-.5em}
	\caption{Optimal SoC trajectories remain within capacity bounds throughout this run.  The peak generation for the optimal solution is about 8GW less than what would be required without load control,  and the generation ramping is reduced significantly.}  
	\vspace{-.75em}
	\label{f:soc} 
\end{figure}

\section{Numerical Simulations}
\label{s:numerics}

Simulations were conducted to validate the main results of this paper.  A discrete time approximation of the resource allocation problem \eqref{qp19} was solved with $5$ classes of loads: ACs, residential WHs with faster time cycles (fwh), commercial WHs with slower time cycles (swh), fridges, and pool pumps (pp),  based on the SoC model used in \cite{cammatkiebusmey18}.  The net-load  $\bfmell$ is based on ``duck curve" predictions for California in March, 2020; this data is obtained from CAISO.

The results in \Cref{f:idiotsprice} were obtained using quadratic cost on the SoC.   Alternatives were tested to avoid the SoC violations (i.e., SoC/Capacity $> 1$) observed in that experiment.   In the results surveyed here,  the cost functions are strongly convex polynomials,   $c_i(x) = \kappa_1 (x/C_i)^8 + \kappa_2 (x/C_i)^2$, where $C_i$ is the energy capacity of the load class $i$ in GWh. For each class of TCLs, $\kappa_1 = 1$ and $\kappa_2 = 0.1$.    
Because QoS requirements for pools are less critical, the
 quadratic cost is maintained for pool pumps: $\kappa_1 = 0$ and $\kappa_2 = 1$. 
Following \cite{cammatkiebusmey18}, the cost on generation is of the form $c_g(x) = \kappa_g [x-\barell]^2$, where $\kappa_g$ is a constant gain.  Table I provides values of the SoC leakage parameters $\alpha_i$ for the different load classes along with the energy capacities $C_i$. The numbers are based on \cite{cammatkiebusmey18} and \cite{matdyscal15}. 

%\spm{out of place at end of this paragraph:
%The polynomial cost functions are designed to indirectly impose capacity constraints on the different TCL classes. For pool pumps, capacity constraints are not hard-set; hence, a quadratic cost is reasonable.
%\\
%I replaced with a brief comment
%\\
%jm:OK
%} 

\begin{center}
	\nobreak
	{\sc Table I:  Load Parameters} \nobreak
	\\[.75em]
	{
		\small
		\begin{tabular}{|| l c c c c c r ||}
			Par. & Unit & ACs & fWHs & sWHs & RFGs & PPs \\
			$\alpha_i$ & hours\textsuperscript{-1} & 0.25 & 0.04 & 0.01 & 0.10 & 0.004  \\
			$C_i$ & GWh &  4 & 2 & 5 & 0.5 & 2  \\
		\end{tabular}
	}
\end{center}

\smallbreak

The top half of \Cref{f:soc} shows the net-load  $\bfmell$ (duck curve), the optimal traditional generation $\bfmg^\star$, and the equilibrium price $\bfvarrho^\star$ (normalized to $\pm 1~ \$/\text{unit energy}$), while the latter half shows the optimal SoC trajectories normalized by the respective energy capacities, $\bfmx^\star_i/C_i$. 

There is remarkable correspondence between the net-load and the equilibrium price signal. As  expected from combining \Cref{r:lambda} with  \Cref{p:rhoAndLambda}, it is observed that
the optimal SoC evolves in tandem with the price signal. For example, the negative prices in the afternoon lead to hotter than average WHs and cooler than average houses, whereas the higher prices in the late evening  result in colder WHs and hotter houses. The polynomial costs on SoC indirectly impose QoS: notice that $\bfmx^\star_i/C_i$ is between $\pm 1$, which implies that the SoC for each load class is within the energy capacity limit.

The optimal SoC trajectories evolve in a two-dimensional subspace. Consequently, 
the optimal SoC trajectory of any load class can be recovered based on observations of the SoC for two other load classes.   
In particular, given the optimal SoC of residential water heaters and ACs,   
we can recover the functions $\{\lambda^\star(t), \ddt \lambda^\star(t): 0 < t \le \clT\}$ as follows:
\begin{equation*}
\begin{bmatrix}
\lambda^\star(t) \\
\ddt \lambda^\star(t)
\end{bmatrix}
=
\begin{bmatrix}
\alpha_{ac}  & -1 \\
\alpha_{fwh} & -1
\end{bmatrix}
^{-1}
\begin{bmatrix}
c'_i(x_{ac}^\star(t)) \\
c'_i(x_{fwh}^\star(t)) 
\end{bmatrix}
\end{equation*}
We can hence recover any load trajectories using \eqref{e:xstar}.

\Cref{f:rwh} shows that the  SoC trajectory of pool pumps recovered using the optimal SoC trajectories of ACs and residential water heaters matches the  optimal SoC trajectory of pool pumps.

\begin{figure}[h]
	\vspace{0.4em}
	\Ebox{.8}{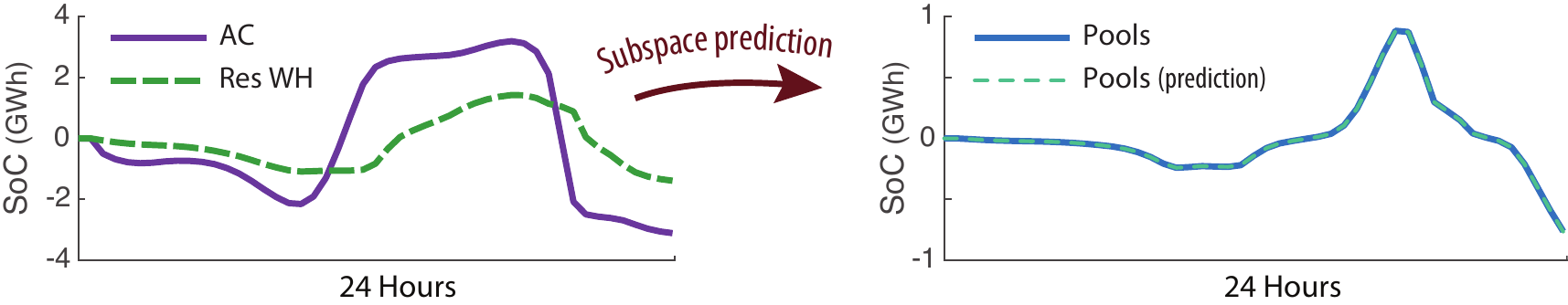}
	\vspace{-.5em}
	\caption{SoC  for pool pumps recovered using those of ACs and   WHs.}
		\vspace{-.75em}
	\label{f:rwh} 
\end{figure}
  
\section{Conclusions}
\label{s:con}

%\spm{too much repetition:
%This paper analyzes the structure of the equilibrium solution to the social planner's problem designed to incorporate load control. Unlike prior research on demand response that imposes costs on the power consumed by the consumer, the cost imposed in this paper is intended to satisfy consumer QoS. }

It is a remarkable fact that a $2M$-dimensional optimal control problem can be reduced to just two dimensions, regardless of the number of load classes $M$.   Beyond its computational value, the result provides new insight and has potential economic implications.

The most valuable implication of \Cref{t:main} is for numerical computation.  We are investigating how to best apply state space collapse.
%\spm{ \bl{Yawn:  \textit{Future research will consider extensions that take into account hard constraints on generation and QoS. }
% I think a penalty function approach is fine.  Let's delete this}}

Analysis of the resource allocation problem for demand dispatch in a stochastic control setting, considering the impact of model uncertainty, is also a topic for future work.

We leave the reader to think over the following: can insight from the economic analysis be used to formulate contracts for grid services with large aggregators and industrial customers? The prices would not be used for real-time control, but they could be utilized to construct metrics in order to evaluate the performance of the participating consumers.

\smallbreak

%\spm{may be useful to put this in the intro in our reference to Caltech: The SPP formulation requires load aggregators that respond to price signals; the presence of aggregators is required to ensure that the power deviation $\bfmz$ takes values in a convex, compact control set $U$. The convexity does not hold for an individual ON/OFF load; we conjecture that there is no equilibrium for demand-response with individual loads.}

%\addtolength{\textheight}{-12cm}  

%%%%%%%%%%%%%%%%%%%%%%%%%%%%%%%%%%%%%%%%%%%%%%%%%%%%%%%%%%%%%%%%%%%%%%%%%%%%%%%%

%%%%%%%%%%%%%%%%%%%%%%%%%%%%%%%%%%%%%%%%%%%%%%%%%%%%%%%%%%%%%%%%%%%%%%%%%%%%%%%%

%%%%%%%%%%%%%%%%%%%%%%%%%%%%%%%%%%%%%%%%%%%%%%%%%%%%%%%%%%%%%%%%%%%%%%%%%%%%%%%%

%\section*{ACKNOWLEDGMENT}
%\bigskip

%\appendix
% IEEE uses   \appendices  but this takes too much space.

\section{Appendix}

We begin with further clarity on cheap control:
\begin{lemma}
\label[lemma]{t:cheap} 
%\label{t:cheap} 

For given $t_0\ge 0$, suppose that $(x(t_0),z(t_0))$ and $(x^+,z^+)$ are two state values satisfying  $x_\sigma(t_0) =x^+_\sigma$,  $z_\sigma(t_0) =z^+_\sigma$.    

Then, for each $\delta>0$, there is a $C^\infty$ input $u$ satisfying $u_\sigma(t)=0$ for all $t\ge t_0$,  and the resulting state trajectory   from $(x(t_0),z(t_0)) $ satisfies
$z(t_0+\delta) = z^+$,   $x(t_0+\delta) = x^+ + O(\delta)$,   and $x(t)$ is bounded on $[t_0,t_0+\delta]$.   
\end{lemma}

\begin{proof}
Without loss of generality we take $t_0=0$.  
Let $f\colon\Re\to\Re_+$ be a $C^\infty$ probability density, 
with support on the interval $(0,\delta)$,
and choose
\[
u_i(t) = [ z_i^+ - z_i(0)]  f(t)   -   [x^+_i - x_i(0) ] f'(t)
\]
where $f'$ denotes the derivative of $f$.
This is a ``cheap control'',  since $u_\sigma(t)=0$ for all $t$.
We then have by definition
\[
\begin{aligned}
z_i(t) &= 
%z_i(0) + \int_0^t u_i(\tau)\, d\tau      
%\\
%&=
z_i(0) + [ z_i^+ - z_i(0)]    \int_0^t  f(\tau)\, d\tau -   [x^+_i - x_i(0) ] f(t)
\end{aligned} 
\]
This gives $z_i(\delta) = z_i^+$,   and 
%$z_\sigma(t) = z_\sigma(0)$ for all $t$.      The SoC trajectory admits the approximation, for $0\le t\le \delta$,
\[
\begin{aligned}
x_i(t)  &   =   e^{-\alpha_i t }x_i(0)   -   \int_0^t  e^{-\alpha_i (t-\tau)  } z_i(\tau)\, d\tau   
\\
&    =  x_i(0)     -   \int_0^t   z_i(\tau)\, d\tau   +O(\delta)
\\
&    =  x_i(0)   + [x^+_i - x_i(0) ]     \int_0^t  f(\tau)\, d\tau   +O(\delta)
\end{aligned} 
\]
The SoC trajectory is bounded,   and $x_i(\delta) = x^+_i + O(\delta)$.  
\end{proof}

\ProofOf{Proof of \Cref{p:JeqK}}
It is obvious that $J^\star$ is finite valued.   To see that it is convex,  let $\{ (x^i,z^i) : i=0,1\}$ denote two initial conditions (starting at time $t_0$),   fix $\theta\in (0,1)$, and denote $(x^\theta,z^\theta) = (1-\theta) (x^0,z^0) +\theta (x^1,z^1)$.    It remains to show that $J^\star(x^\theta,z^\theta,t_0)  \le (1-\theta) J^\star(x^0,z^0,t_0)  +\theta J^\star(x^1,z^1,t_0)$ for each $t_0$.  Consider any continuous input-state trajectories:
\[
\{ (u^i_{[t_0,\clT]} , x^i_{[t_0,\clT]} ,z^i_{[t_0,\clT]} ) : i=0,1\}
\]
with given initial conditions $x^i(t_0)=x^i$,  $z^i(t_0)=z^i$.   Because the system is linear, it follows that the convex combination is feasible from $(x^\theta,z^\theta) $:    with $u^\theta_{[t_0,\clT]}$ defined as the convex combination of the inputs,  the resulting state trajectory is the convex combination $(x^\theta_{[t_0,\clT]} ,z^\theta_{[t_0,\clT]} ) $.    Consequently,
\[
\begin{aligned}
J^\star(x^\theta,z^\theta, t_0) &\le  \int_{t_0}^{\clT}     c(x^\theta(t), z^\theta(t), u^\theta(t),t) \, dt
\\
&
\le  (1-\theta) \int_{t_0}^{\clT}     c(x^0(t), z^0(t), u^0(t),t) \, dt   +   \theta \int_{t_0}^{\clT}     c(x^1(t), z^1(t), u^1(t)) \, dt  ,
\end{aligned} 
\]
where the first inequality is the definition of $J^\star$ as an infimum, and the second follows from convexity of the cost function.
The proof of (i) is completed on
taking the infimum over $u^i_{[t_0,\clT]} $ for each $  i=0,1$.

%\spm{This long proof seems silly.   At least you can see why I felt it was obvious. }

\smallbreak

\def\tilx{\tilde x}
\def\tilu{\tilde u}

We next prove (ii).  It is clear from the definitions that $K^\star(x_\sigma,z_\sigma,t_0) \le J^\star(x,z,t_0)$;    we establish next the reverse inequality, for each $x,z,t_0$.   For  $\delta>0$ fixed,   the following pair of bounds will be established:
\[
\begin{aligned}
K^\star(x_\sigma,z_\sigma,t_0+\delta)& \le   
K^\star(x_\sigma,z_\sigma,t_0)    + O(\delta) ,
\\
J^\star(x,z,t_0)  & \le  K^\star(x_\sigma,z_\sigma,t_0+\delta)  + O(\delta) .
\end{aligned} 
\]
Since $\delta>0$ is arbitrary, these bounds are sufficient to establish (ii). 
The first inequality follows because the cost function is non-negative, so    only the second   requires proof.  
 
Let $ x^\star,  z^\star\in\Re^M$ denote the  optimizers in \eqref{e:V2}, so that in particular  $ x^\star_\sigma= x_\sigma$,  $z^\star_\sigma= z_\sigma$.   
Let $u$ denote the input described in \Cref{t:cheap} with $x^+=x^\star$,  $z^+=z^\star$.   
The resulting state trajectory satisfies the conclusions of  \Cref{t:cheap},   so that in particular $x^\delta(t_0+\delta) = x^\star +O(\delta)$.
We thus obtain the desired bound:
\[
\begin{aligned}
J^\star(x,z,t_0)   & \le     \int_{t_0}^{t_0+\delta}     c(x^\delta(t), z^\delta(t), u^\delta(t)) \, dt  + J^\star(x^\delta(t_0+\delta), z^\delta(t_0+\delta),t_0+\delta)  
\\
& = K^\star(x_\sigma,z_\sigma,t_0+\delta)  + O(\delta).
\end{aligned} 
\]
The first inequality is due to Bellman's principle of optimality, while the second approximation is a consequence of the following:
 (i) the first cost term is bounded by $O(\delta)$ as a consequence of the cheap control input,
 (ii) $J^\star$ is Lipschitz with respect to the state variables (as it is convex), 
 and (iii) the definition of $K^\star$ in  \eqref{e:V2}.
\qedIEEE

\ProofOf{Co-state dynamics}

The dynamics of the dual variables  appearing in \Cref{r:lambda} are given in the following lemma.

 \begin{lemma}
	\label[lemma]{l:lambdai} 
	Let $t_0 \in (0,\clT)$. The optimal input-state  $(\bfnu^\star, \bfmx^\star, \bfmz^\star)$ and dual variables  $\{\lambda^\star,\beta^\star\}$ satisfy the following co-state equations:
	\begin{subequations}
	\begin{align}
	\label{e:Lambdai}
	\ddt \lambda^\star_i(t) &= \alpha_i \lambda^\star_i(t) - c'_i (x^\star(t)),\\
	\label{e:betai}
	\ddt \beta^\star_i(t) &= \lambda^\star_i(t) + c'_g(\ell(t) - z^\star_\sigma(t)), \quad t \in [t_0, \clT] ,
	\end{align}
	\end{subequations}
	with   boundary condition $\lambda^\star_i(\clT) = 0, ~ \beta^\star_i(\clT) = 0$ for each $i$.
\end{lemma}

\begin{proof}
    The state $ (x^\star(t), z^\star(t) )$ is continuously differentiable  and the control $u^\star(t)$ is continuous on $[t_0, \clT]$  as a consequence of (A2). In addition, the dynamics, \Cref{e:soccons,e:loadrampcons}, are linear and hence continuously differentiable with respect to each of the variables. Further, (A1)  implies that the cost function $\bfmath{c}$ in \eqref{e:modCost} is continuously differentiable with respect to all the variables. Consequently, the optimal input-state pair $\bfnu^\star, (\bfmx^\star, \bfmz^\star)$ satisfies Pontryagin's minimum principle on the closed interval $[t_0, \clT]$   \cite[Section 4.2]{lib11}.  The rest of the proof follows from this result. In particular, the   minimum principle implies
\[
\begin{aligned}
\ddt \lambda^\star_i(t)& = - \ddxp H \bigl(x^\star(t), z^\star(t),u^\star(t), \lambda^\star(t),  \beta^\star(t), t) \bigr),
\\
\ddt \beta^\star_i(t) &= - \ddzp H \bigl(x^\star(t), z^\star(t),u^\star(t), \lambda^\star(t),  \beta^\star(t), t) \bigr) 
\end{aligned} 
\]
which
yields \eqref{e:Lambdai}. The boundary conditions  $\lambda^\star_i (\clT)  =  \beta^\star_i (\clT)  = 0$ hold because there is no   terminal cost   \cite[Theorem 1]{vin15}.
 \end{proof}

\ProofOf{Optimal input for the descriptor system}

The next result provides the optimal input $\bfnu^\star_\sigma$ for the optimal descriptor system states $(\bfmx^\star_\sigma, \bfmz^\star_\sigma)$:
\begin{lemma}
	\label[lemma]{l:FB}
	Let $t_0 \in (0,\clT)$.  The optimal input on $[t_0,\clT]$ is
	\begin{equation}
	\label{e:optimalfb}
	u_\sigma^\star(t) = \ddt \ell(t) - \frac{1}{2 \kappa} \beta_i^\star(t), \qquad t \in [t_0, \clT]
	\end{equation}
Consequently,  $\beta^\star(t)\eqdef\beta_i^\star(t)$,   $\lambda^\star(t)\eqdef\lambda_i^\star(t)$ 
are independent of $i$.
\end{lemma}
\begin{proof}
Similar to the proof of \Cref{l:lambdai}, the conclusion   \eqref{e:optimalfb} is a consequence of Pontryagin's minimum principle, and the first-order condition for minimality:
\[
0=	\ddnup H \bigl(x^\star(t), z^\star(t),u^\star(t), {\lambda^\star}(t),  {\beta^\star}(t), t) \bigr)  
\]
This establishes \eqref{e:optimalfb}, which then implies that    $\beta_i^\star(t)  $  is independent of $i$ for each $t$.    
The conclusion that $ \lambda_i^\star(t)$ is also independent of $i$ follows from  \cref{e:betai}. 
	 \end{proof}

%\begin{remark}[Co-state collapse]
%	\label{r:costatecollapse}   	\qedIEEE
%\end{remark}

The lemma reinforces the co-state collapse identified in \Cref{r:lambda}. 
The proof of the latter is provided next.

\ProofOf{Proof of \Cref{r:lambda}}
The left hand equalities in \eqref{e:lambdas} are a familiar result:  the optimal co-state trajectory is the gradient of the value function with respect to the state variable \cite[Theorem 3.1]{clawin87}. The right-hand equalities are from the definition \eqref{e:lambda}. 
\qedIEEE

\ProofOf{Proof of \Cref{t:main}} 
Equations \eqref{e:SoC} and \eqref{e:contro} are the state equations.
Equations \eqref{e:2dsub} and \eqref{e:betamain} and the final time boundary conditions on the co-state variables follow from \Cref{l:lambdai,l:FB}.
%,r:costatecollapse
The optimal feedback policy   \eqref{e:rampsum} is obtained from \Cref{l:FB}.
The initial state mapping $\bigl(x(0+), z(0+)\bigr)$ is given by \eqref{e:IC}.  

 As a consequence of (A2) and the fact that the cost functional is strictly convex with respect to the control and that the dynamics are linear, the solution $ (\bfmx^\star, \bfmz^\star, \bflambda^\star, \bfbeta^\star, \bfnu^\star  )$ satisfying Pontryagin's minimum principle is both necessary and sufficient for optimality \cite[Chapter 7]{brepic07}.
 %\spm{Do we have a reference (Sept. 7, 2019). \rd{JM: (A2) assumes a unique solution, so it is sufficient, right? There is a discussion on the sufficiency of PMP for convex functionals in Ch7 of Intro to Mathematical Theory of Control by Bresnan and Piccoli, but I do not have access to the book in Austin.}}
 \qedIEEE

\ProofOf{Proof of \Cref{p:conrho}}  
	As $\varrho \in C^1$, it follows from assumption (A1) and \eqref{e:genCost} that $\clL_g, \clL_i \in C^1$. Moreover, $(g^\varrho(t), \dot g^\varrho(t))$ and $(x^\varrho_i(t), \dot x^\varrho_i(t))$ are continuous on the half-open interval $(0, \clT]$. Consequently, the Euler-Lagrange equations form the necessary first-order conditions for weak extrema \cite[Section 2.3.3]{lib11}. 
%	\spm{I'd rather have a theorem citation rather than a section! \rd{JM: unfortunately Liberzon doesn't have a theorem-proof format, but sections and derivations.  } }
The solution to the minimization problem in \eqref{e:grid} and \eqref{e:aggregator}  at the stationary minimum are the following Euler-Lagrange equations,
\begin{align*}
	\FRAC{1}{\partial}{\partial g} \clL_g (g^\varrho, \dot g^\varrho, t) - \ddt \FRAC{1}{\partial}{\partial \dot g} \clL_g (g^\varrho, \dot g^\varrho, t) & = 0, 
\\
	\FRAC{1}{\partial}{\partial x_i} \clL_i (x^\varrho_i, \dot x^\varrho_i, t) - \ddt \FRAC{1}{\partial}{\partial \dot x_i} \clL_i (x^\varrho_i, \dot x^\varrho_i, t) & = 0,
\end{align*}
which result in \eqref{e:ELg} and \eqref{e:EL}, respectively. The terminal-time boundary conditions are obtained by, respectively, setting $\FRAC{1}{\partial}{\partial \dot g} \clL_g (g^\varrho, \dot g^\varrho, t)|_{t=\clT} = 0$ and $\FRAC{1}{\partial}{\partial \dot x_i} \clL_i (x^\varrho_i, \dot x^\varrho_i, t)|_{t=\clT} = 0$ \cite[Section 2.3.5]{lib11}.
\qedIEEE

\ProofOf{Proof of \Cref{p:rhoAndLambda}}
With $t \in (0, \clT]$, setting $\varrho(t) = -\lambda^\star(t)$ in \eqref{e:EL} and comparing with \eqref{e:2dsub} yields $x_i^\varrho(t) = x_i^\star(t)$ for each $i$, which is the optimal solution to the primal problem \eqref{qp19}. This implies that there is no duality gap: $-\bflambda^\star$ maximizes the dual functional $\phi^\star$. 
\qedIEEE

\ProofOf{Proof of \Cref{p:avgprice}}
The proof follows from integrating \eqref{e:EL} and \eqref{e:ELg} over $(0, \clT]$ and dividing throughout by $\clT$. In particular, the error terms follow from the fundamental theorem of calculus.
\qedIEEE

 \def\urls#1{{\scriptsize\url{#1}}}
 
\bibliographystyle{IEEEtran}
%\bibliography{IEEEabrv,strings,markov,q,CollapseExtras.bib,MoyeCDCRefs}

\begin{thebibliography}{10}
%\providecommand{\url}[1]{\scriptsize #1}    %% Sean modified definition
\csname url@samestyle\endcsname
\providecommand{\newblock}{\relax}
\providecommand{\bibinfo}[2]{#2}
\providecommand{\BIBentrySTDinterwordspacing}{\spaceskip=0pt\relax}
\providecommand{\BIBentryALTinterwordstretchfactor}{4}
\providecommand{\BIBentryALTinterwordspacing}{\spaceskip=\fontdimen2\font plus
\BIBentryALTinterwordstretchfactor\fontdimen3\font minus
  \fontdimen4\font\relax}
\providecommand{\BIBforeignlanguage}[2]{{%
\expandafter\ifx\csname l@#1\endcsname\relax
\typeout{** WARNING: IEEEtran.bst: No hyphenation pattern has been}%
\typeout{** loaded for the language `#1'. Using the pattern for}%
\typeout{** the default language instead.}%
\else
\language=\csname l@#1\endcsname
\fi
#2}}
\providecommand{\BIBdecl}{\relax}
\BIBdecl

\bibitem{matmoymey19}
J.~Mathias, R.~Moye, S.~Meyn, and J.~Warrington.
  ``State space collapse in resource allocation for demand dispatch,''
  in \emph{ Proc. of the IEEE Conf. on Dec. and Control}, Dec.
  2019.



\bibitem{cammatkiebusmey18}
N.~Cammardella, J.~Mathias, M.~Kiener, A.~Bu{\v s}i{\'c}, and S.~Meyn,
  ``Balancing {California's} grid without batteries,'' in \emph{Proc. of the
  IEEE Conf. on Dec. and Control}, Dec 2018, pp. 7314--7321.

\bibitem{chehasmatbusmey18}
\BIBentryALTinterwordspacing
Y.~Chen, M.~U. Hashmi, J.~Mathias, A.~Bu{\v{s}}i{\'{c}}, and S.~Meyn,
  ``Distributed control design for balancing the grid using flexible loads,''
  in \emph{Energy Markets and Responsive Grids: Modeling, Control, and
  Optimization}, S.~Meyn, T.~Samad, I.~Hiskens, and J.~Stoustrup, Eds.\hskip
  1em plus 0.5em minus 0.4em\relax New York, NY: Springer, 2018, pp. 383--411.
  
%  [Online]. Available: \urls{https://doi.org/10.1007/978-1-4939-7822-9_16}
%\BIBentrySTDinterwordspacing

\bibitem{cheche17b}
M.~Chertkov and V.~Y. Chernyak, ``Ensemble control of cycling energy loads:
  {Markov Decision Approach},'' in \emph{Energy Markets and Responsive Grids: Modeling, Control, and
  	Optimization}, S.~Meyn, T.~Samad, I.~Hiskens, and J.~Stoustrup, Eds.\hskip
  1em plus 0.5em minus 0.4em\relax New York, NY: Springer, 2018.

\bibitem{almesphinfropauami18}
\BIBentryALTinterwordspacing
M.~Almassalkhi, L.~D. Espinosa, P.~D. H.~Hines, J.~Frolik, S.~Paudyal, and
  M.~Amini, ``Asynchronous coordination of distributed energy resources with
  packetized energy management,'' in \emph{Energy Markets and Responsive Grids:
  Modeling, Control, and Optimization}, S.~Meyn, T.~Samad, I.~Hiskens, and
  J.~Stoustrup, Eds.\hskip 1em plus 0.5em minus 0.4em\relax New York, NY:
  Springer, 2018, pp. 333--361.
  
  % [Online]. Available: \url{https://doi.org/10.1007/978-1-4939-7822-9_14}\BIBentrySTDinterwordspacing


\bibitem{matbusmey17}
J.~{Mathias}, A.~{Bu{\v s}i{\'c}}, and S.~{Meyn}, ``Demand dispatch with
  heterogeneous intelligent loads,'' in \emph{Proc. {50th Annual Hawaii
  International Conference on System Sciences (HICSS), and arXiv 1610.00813}},
  2017.

\bibitem{matkadbusmey16}
J.~Mathias, R.~Kaddah, A.~Bu\v{s}i\'{c}, and S.~Meyn, ``Smart fridge / dumb
  grid? {Demand Dispatch} for the power grid of 2020,'' in \emph{Proc. {49th
  Annual Hawaii International Conference on System Sciences (HICSS)}}, Jan
  2016, pp. 2498--2507.

\bibitem{trotinstr16}
V.~Trovato, S.~H. Tindemans, and G.~Strbac, ``Leaky storage model for optimal
  multi-service allocation of thermostatic loads,'' \emph{IET Generation,
  Transmission \& Distribution}, vol.~10,  pp. 585--593, 2016.

\bibitem{haosanpoovin15}
H.~Hao, B.~M. Sanandaji, K.~Poolla, and T.~L. Vincent, ``Aggregate flexibility
  of thermostatically controlled loads,'' \emph{IEEE Trans. on Power Systems},
  vol.~30, no.~1, pp. 189--198, Jan 2015.

\bibitem{chebusmey14}
Y.~Chen, A.~Bu\v{s}i\'{c}, and S.~Meyn, ``Individual risk in mean field control
  with application to automated demand response,'' in \emph{Proc. of the IEEE
  Conf. on Dec. and Control}, Dec 2014, pp. 6425--6432.

\bibitem{meybarbusyueehr15}
S.~Meyn, P.~Barooah, A.~Bu\v{s}i\'{c}, Y.~Chen, and J.~Ehren, ``Ancillary
  service to the grid using intelligent deferrable loads,'' \emph{IEEE Trans.
  Automat. Control}, vol.~60, no.~11, pp. 2847--2862, Nov 2015.

\bibitem{rei84b}
M.~I. Reiman, ``Some diffusion approximations with state space collapse,'' in
  \emph{Modelling and Performance Evaluation Methodology}, F.~Baccelli and
  G.~Fayolle, Eds.\hskip 1em plus 0.5em minus 0.4em\relax Berlin, Heidelberg:
  Springer Berlin Heidelberg, 1984, pp. 207--240.

\bibitem{CTCN}
S.~P. Meyn, \emph{Control Techniques for Complex Networks}.\hskip 1em plus
  0.5em minus 0.4em\relax Cambridge University Press, 2007, pre-publication
  edition available online.

\bibitem{sakorekok84}
\BIBentryALTinterwordspacing
V.~Saksena, J.~O'Reilly, and P.~Kokotovic, ``Singular perturbations and
  time-scale methods in control theory: Survey 1976--1983,'' \emph{Automatica},
  vol.~20, no.~3, pp. 273 -- 293, 1984.
  
%   [Online]. Available:
%  \urls{http://www.sciencedirect.com/science/article/pii/000510988490044X}
%\BIBentrySTDinterwordspacing

\bibitem{fra79}
B.~Francis, ``The optimal linear-quadratic time-invariant regulator with cheap
  control,'' \emph{IEEE Trans. Automat. Control}, vol.~24, no.~4, pp. 616--621,
  August 1979.

\bibitem{hausil83}
M.~L. Hautus and L.~M. Silverman, ``System structure and singular control,''
  \emph{Linear algebra and its applications},   pp. 369--402, 1983.

\bibitem{chomey07}
I.-K. Cho and S.~P. Meyn, ``Efficiency and marginal cost pricing in dynamic
competitive markets with friction,'' in \emph{Proc. of the IEEE Conf. on Dec.
	and Control}, Dec. 2007, pp. 771--778.

\bibitem{chomey10}
I.-K. Cho and S.~P. Meyn, ``Efficiency and marginal cost pricing in dynamic
  competitive markets with friction,'' \emph{Theoretical Economics}, vol.~5,
  no.~2, pp. 215--239, 2010.

\bibitem{wankownegshamey10}
G.~Wang, A.~Kowli, M.~Negrete-Pincetic, E.~Shafieepoorfard, and S.~Meyn, ``A
  control theorist's perspective on dynamic competitive equilibria in
  electricity markets,'' in \emph{Proc. 18th World Congress of the
  International Federation of Automatic Control (IFAC)},   2011.

\bibitem{kizman10b}
A.~Kizilkale and S.~Mannor, ``Volatility and efficiency in markets with
  friction,'' in \emph{48th Annual Allerton Conference on Communication,
  Control, and Computing}, 2010, pp. 50--57.

\bibitem{zavani11a}
V.~M. Zavala and M.~Anitescu, ``On the dynamic stability of electricity
  markets,'' Argonne National Laboratory, Preprint ANL/MCS-P1834-0111, Tech.
  Rep., 2011.

\bibitem{lichelow11}
N.~Li, L.~Chen, and S.~H. Low, ``Optimal demand response based on utility
  maximization in power networks,'' in \emph{IEEE Power and Energy Society
  General Meeting}, July 2011, pp. 1--8.

\bibitem{chebusmey17a}
Y.~Chen, A.~Bu\v{s}i\'{c}, and S.~Meyn, ``State estimation for the individual
  and the population in mean field control with application to demand
  dispatch,'' \emph{IEEE Transactions on Automatic Control}, vol.~62, no.~3,
  pp. 1138--1149, March 2017.

\bibitem{park1987}
R.~E. Park and B.~M. Mitchell, ``{Optimal Peak-Load Pricing for Local Telephone
  Calls},'' \emph{{RAND Publication Series}}, 1987.

\bibitem{lobluhinmey19}
\BIBentryALTinterwordspacing
H.~Lo, S.~Blumsack, P.~Hines, and S.~Meyn, ``Electricity rates for the zero
  marginal cost grid,'' \emph{The Electricity Journal}, vol.~32, no.~3, pp. 39
  -- 43, 2019. [Online].  
  \urls{http://www.sciencedirect.com/science/article/pii/S1040619019300594}
\BIBentrySTDinterwordspacing

\bibitem{tak85}
A.~Takayama, \emph{Mathematical economics}. Cambridge Univ.\ Press, 1985.

\bibitem{matdyscal15}
\BIBentryALTinterwordspacing
J.~L. Mathieu, M.~E. Dyson, and D.~S. Callaway, ``Resource and revenue
  potential of {California} residential load participation in ancillary
  services,'' \emph{Energy Policy}, vol.~80,   pp. 76 -- 87, 2015.
%  [Online].  
%  \urls{http://www.sciencedirect.com/science/article/pii/S0301421515000427}
%%\BIBentrySTDinterwordspacing

\bibitem{lib11}
D.~Liberzon, \emph{Calculus of variations and optimal control theory: a concise
  introduction}.\hskip 1em plus 0.5em minus 0.4em\relax Princeton University
  Press, 2011.

\bibitem{vin15}
R.~B. Vinter, ``Optimal control and {Pontryagin's Maximum Principle},''
  \emph{Encyclopedia of Systems and Control}, pp. 950--956, 2015.

\bibitem{clawin87}
F.~H. Clarke and R.~B. Vinter, ``The relationship between the maximum principle
  and dynamic programming,'' \emph{SIAM Journal on Control and Optimization},
  vol.~25, no.~5, pp. 1291--1311, 1987.
  
\bibitem{brepic07}
A.~Bressan and B.~Piccoli, \emph{Introduction to the mathematical theory of control}.\hskip 1em plus 0.5em minus 0.4em\relax American Institute of Math.\ Sciences Springfield, 2007.

\end{thebibliography}

\def\cprime{$'$}\def\cprime{$'$}

{}
\null  
\null  %Needed with \usepackage{flushend}

\end{document}